\documentclass{llncs}

\usepackage{amsmath, amssymb}
\usepackage{algorithm}
\usepackage{algorithmic}
\usepackage{color, array, colortbl}
\usepackage{graphicx}
\usepackage{multirow}
\usepackage{framed,color}
\definecolor{shadecolor}{rgb}{0.9,0.9,0.9}
\newcommand{\mathsym}[1]{{}}

\newtheorem{fact}[theorem]{Fact}

\newcommand{\cancel}[1]{}

\newcommand{\N}{\ensuremath{\mathbb{N}}}

\newcommand{\E}{\ensuremath{\mathbb{E}}}
\renewcommand{\Pr}{\ensuremath{\mathbb{P}}}

\begin{document}

\title{A Jamming-Resistant MAC Protocol for\\Multi-Hop Wireless
Networks\thanks{A preliminary version of this article appeared at
the 24th International Symposium on Distributed Computing (DISC),
2010.}}


\author {
   Andrea Richa$^1$, Christian Scheideler$^2$, Stefan Schmid$^3$, Jin Zhang$^1$\\
   \small $^1$ Computer Science and Engineering, SCIDSE,  Arizona State University\\
       \small Tempe, AZ 85287, USA; \{aricha,jzhang82\}@asu.edu\\
  \small $^2$ Department of Computer Science, University of
Paderborn, D-33102 Paderborn,
  Germany,; \small scheideler@upb.de\\
 \small $^3$ Deutsche Telekom Laboratories, TU Berlin, D-10587 Berlin,
  Germany\\ \small stefan@net.t-labs.tu-berlin.de\
}

\institute{ }

\date{}

\maketitle
%
%
%

\sloppy

\begin{abstract}
This paper presents a simple local medium access control protocol, called
\textsc{Jade}, for multi-hop wireless networks with a single channel that is
provably robust against adaptive adversarial jamming. The wireless network is
modeled as a unit disk graph on a set of nodes distributed arbitrarily in the
plane. In addition to these nodes, there are adversarial jammers that know the
protocol and its entire history and that are allowed to jam the wireless
channel at any node for an arbitrary $(1-\epsilon)$-fraction of the time
steps, where $0<\epsilon<1$ is an arbitrary constant. We assume that the nodes
cannot distinguish between jammed transmissions and collisions of regular
messages. Nevertheless, we show that \textsc{Jade} achieves an asymptotically
optimal throughput if there is a sufficiently dense distribution of nodes.
\end{abstract}


\section{Introduction}\label{sec:introduction}

The problem of coordinating the access to a shared medium is a central
challenge in wireless networks. In order to solve this problem, a proper
medium access control (MAC) protocol is needed. Ideally, such a protocol
should not only be able to use the wireless medium as effectively as possible,
but it should also be robust against attacks. Unfortunately, most of the MAC
protocols today can be easily attacked. A particularly critical class of
attacks are {\em jamming attacks} (i.e., denial-of-service attacks on the
broadcast medium). Jamming attacks are typically easy to implement as the
attacker does not need any special hardware. Attacks of this kind usually aim
at the physical layer and are realized by means of a high transmission power
signal that corrupts a communication link or an area, but they may also occur
at the MAC layer, where an adversary may either corrupt control packets or
reserve the channel for the maximum allowable number of slots so that other
nodes experience low throughput by not being able to access the channel. In
this paper we focus on jamming attacks at the physical layer, that is, the
interference caused by the jammer will not allow the nodes to receive
messages. The fundamental question that we are investigating is: {\em Is there
a MAC protocol such that for any physical-layer jamming strategy, the protocol
will still be able to achieve an asymptotically optimal throughput for the
non-jammed time steps?} Such a protocol would {\em force} the jammer to jam
all the time in order to prevent any successful message transmissions. Finding
such a MAC protocol is not a trivial problem. In fact, the widely used IEEE
802.11 MAC protocol already fails to deliver any messages for very simple
oblivious jammers that jam only a small fraction of the time
steps~\cite{Raj08}. On the positive side, Awerbuch et al.~\cite{singlehop08}
have demonstrated that there are MAC protocols which are provably robust
against even massive adaptive jamming, but their results only hold for
single-hop wireless networks with a single jammer, that is, all nodes
experience the same jamming sequence.

In this paper, we significantly extend the results in~\cite{singlehop08}. We
present a MAC protocol called $\textsc{Jade}$ (a short form of ``jamming
defense'') that can achieve a constant fraction of the best possible
throughput for a large class of jamming strategies in a large class of
multi-hop networks where transmissions and interference can be modeled using
unit-disk graphs. These jamming strategies include jamming patterns that can
be completely different from node to node. It turns out that while
$\textsc{Jade}$ differs only slightly from the MAC protocol
of~\cite{singlehop08}, the proof techniques needed for the multi-hop setting
significantly differ from the techniques in~\cite{singlehop08}.

\subsection{Model} \label{sec:model}

We consider the problem of designing a robust MAC protocol for multi-hop
wireless networks with a single wireless channel. The wireless network is
modeled as a \emph{unit disk graph} (UDG) $G=(V,E)$ where $V$ represents a set
of $n=|V|$ honest and reliable nodes and two nodes $u,v\in V$ are within each
other's transmission range, i.e., $\{u,v\}\in E$, if and only if their
(normalized) distance is at most 1. We assume that time proceeds in
synchronous time steps called {\em rounds}. In each round, a node may either
transmit a message or sense the channel, but it cannot do both. A node which
is sensing the channel may either $(i)$ sense an {\em idle} channel (if no
other node in its transmission range is transmitting at that round and its
channel is not jammed), $(ii)$ sense a {\em busy} channel (if two or more
nodes in its transmission range transmit at that round or its channel is
jammed), or $(iii)$ {\em receive} a packet (if exactly one node in its
transmission range transmits at that round and its channel is not jammed).

In addition to these nodes there is an adversary (who may control any number
of jamming devices). We allow the adversary to know the protocol and its
entire history and to use this knowledge in order to jam the wireless channel
at will at any round (i.e, the adversary is {\em adaptive}). However, like in
\cite{singlehop08}, the adversary has to make a jamming decision \emph{before}
it knows the actions of the nodes at the current round. The adversary can jam
the nodes individually at will, as long as for every node $v$, at most a
$(1-\epsilon)$-fraction of its rounds is jammed, where $\epsilon>0$ can be an
arbitrarily small constant. That is, $v$ has the chance to receive a message
in at least an $\epsilon$-fraction of the rounds. More formally, an adversary
is called {\em $(T,1-\epsilon)$-bounded} for some $T \in \N$ and
$0<\epsilon<1$, if for any time window of size $w \ge T$ and at any node $v$,
the adversary can jam at most $(1-\epsilon)w$ of the $w$ rounds at $v$.

Given a node $v$ and a time interval $I$, we define $f_v(I)$ as the number of
time steps in $I$ that are non-jammed at $v$ and $s_v(I)$ as the number of
time steps in $I$ in which $v$ successfully receives a message. A MAC protocol
is called {\em $c$-competitive} against some $(T,1-\epsilon)$-bounded
adversary if, for any time interval $I$ with $|I| \ge K$ for a sufficiently
large $K$ (that may depend on $T$ and $n$),
$
  \sum_{v \in V} s_v(I) \ge c \cdot \sum_{v \in V} f_v(I).
$
In other words, a $c$-competitive MAC protocol can achieve at least a
$c$-fraction of the best possible throughput.

Our goal is to design a {\em symmetric local-control} MAC protocol (i.e.,
there is no central authority controlling the nodes, and all the nodes are
executing the same protocol) that has a constant-competitive throughput
against any $(T,1-\epsilon)$-bounded adversary in any multi-hop network that
can be modeled as a UDG. In order to obtain a more refined picture of the
competitiveness of our protocol, we will also investigate so-called
$k$-uniform adversaries. An adversary is {\em $k$-uniform} if the node set $V$
can be partitioned into $k$ subsets so that the jamming sequence is the same
within each subset. In other words, we require that at all times, the nodes in
a subset are either all jammed or all non-jammed. Thus, a 1-uniform jammer
jams either everybody or nobody in a round whereas an $n$-uniform jammer can
jam the nodes individually at will.


In this paper, we will say that a claim holds \emph{with high probability
(w.h.p.)} iff it holds with probability at least $1-1/n^c$ for any constant $c
\ge 1$; it holds \emph{with moderate probability (w.m.p.)} iff it holds with
probability at least $1-1/(\log n)^c$ for any constant $c \ge 1$.

\subsection{Related Work}

Due to the topic's importance, wireless network jamming has been extensively
studied in the applied research fields
\cite{Alnifie:Q2SWinet07,Brown:mobihoc06,Chiang:mobicom07,citer1,Law:sasn05,Li:infocom07,LiuNST07,NavdaBGR07,NegiP03,ThuenteA06,Wood:secon07,Xu:network06,Xu:mobihoc05},
both from the attacker's perspective
\cite{Chiang:mobicom07,Law:sasn05,Li:infocom07,Xu:mobihoc05} as well as from
the defender's perspective
\cite{Alnifie:Q2SWinet07,Brown:mobihoc06,Chiang:mobicom07,Li:infocom07,LiuNST07,NavdaBGR07,Wood:secon07,Xu:mobihoc05}---also
in multi-hop settings (e.g.~\cite{multi5,multi4,multi3,multi2,multi1}).

Traditionally, jamming defense mechanisms operate on the physical layer
\cite{LiuNST07,NavdaBGR07,SimonOSL01}. Mechanisms have been designed to {\em
avoid} jamming as well as {\em detect} jamming. Spread spectrum technology has
been shown to be very effective to avoid jamming as with widely spread
signals, it becomes harder to detect the start of a packet quickly enough in
order to jam it. Unfortunately, protocols such as IEEE 802.11b use relatively
narrow spreading~\cite{IEEE99}, and some other IEEE 802.11 variants spread
signals by even smaller factors~\cite{Brown:mobihoc06}. Therefore, a jammer
that simultaneously blocks a small number of frequencies renders spread
spectrum techniques useless in this case. As jamming strategies can come in
many different flavors, detecting jamming activities by simple methods based
on signal strength, carrier sensing, or packet delivery ratios has turned out
to be quite difficult \cite{Li:infocom07}.

Recent work has also studied \emph{MAC layer strategies} against jamming,
including coding strategies \cite{Chiang:mobicom07}, channel surfing and
spatial retreat \cite{Alnifie:Q2SWinet07,Xu:wws04}, or mechanisms to hide
messages from a jammer, evade its search, and reduce the impact of corrupted
messages \cite{Wood:secon07}. Unfortunately, these methods do not help against
an adaptive jammer with {\em full} information about the history of the
protocol, like the one considered in our work.

In the theory community, work on MAC protocols has mostly focused on
efficiency. Many of these protocols are random backoff or tournament-based
protocols \cite{Bender05,Chlebus06,Gold00,Hastad96,Kwak05,Raghavan99} that do
not take jamming activity into account and, in fact, are not robust against it
(see \cite{singlehop08} for more details). The same also holds for many MAC
protocols that have been designed in the context of broadcasting \cite{CR06}
and clustering \cite{Kuhn04}.
Also some work on jamming is known (e.g.,~\cite{citer2} for a short overview).
There are two basic approaches in the literature. The first assumes randomly
corrupted messages (e.g.~\cite{PP05}), which is much easier to handle than
adaptive adversarial jamming \cite{Raj08}. The second line of work either
bounds the number of messages that the adversary can transmit or disrupt with
a limited energy budget (e.g.~\cite{GGN06,KBKV06}) or bounds the number of
channels the adversary can jam
(e.g.~\cite{dolevpodc,DGGN07,DGGN08,shlomi07,GGKN09,seth09,dcoss09}).

The protocols in \cite{GGN06,KBKV06} can tackle adversarial jamming at both
the MAC and network layers, where the adversary may not only be jamming the
channel but also introducing malicious (fake) messages (possibly with address
spoofing). However, they depend on the fact that the adversarial jamming
budget is finite, so it is not clear whether the protocols would work under
heavy continuous jamming. (The result in \cite{GGN06} seems to imply that a
jamming rate of $1/2$ is the limit whereas the handshaking mechanisms in
\cite{KBKV06} seem to require an even lower jamming rate.)

In the multi-channel version of the problem introduced in the theory community
by Dolev \cite{shlomi07} and also studied in
\cite{dolevpodc,DGGN07,DGGN08,shlomi07,GGKN09,seth09,dcoss09}, a node can only
access one channel at a time, which results in protocols with a fairly large
runtime (which can be exponential for deterministic protocols
\cite{DGGN07,GGKN09} and at least quadratic in the number of jammed channels
for randomized protocols \cite{DGGN08,dcoss09} if the adversary can jam almost
all channels at a time). Recent work~\cite{dolevpodc} also focuses on the
wireless synchronization problem which requires devices to be activated at
different times on a congested single-hop radio network to synchronize their
round numbering while an adversary can disrupt a certain number of frequencies
per round. Gilbert et al.~\cite{seth09} study robust information exchange in
single-hop networks.

Our work is motivated by the work in \cite{Raj08} and \cite{singlehop08}. In
\cite{Raj08} it is shown that an adaptive jammer can dramatically reduce the
throughput of the standard MAC protocol used in IEEE 802.11 with only limited
energy cost on the adversary side. Awerbuch et al. \cite{singlehop08}
initiated the study of throughput-competitive MAC protocols under continuously
running, adaptive jammers, but they only consider single-hop wireless
networks. We go one step further by considering \emph{multi-hop networks}
where different nodes can have different channel states at a time, e.g., a
transmission may be received only by a fraction of the nodes. It turns out
that while the MAC protocol of~\cite{singlehop08} can be adopted to the
multi-hop setting with a small modification, the proof techniques cannot. We
are not aware of any other theoretical work on MAC protocols for multi-hop
networks with provable performance against adaptive jamming.

\subsection{Our Contributions}

In this paper, we present a robust MAC protocol called {\sc Jade}. {\sc Jade}
is a fairly simple protocol: it is based on a very small set of assumptions
and rules and has a minimal storage overhead. In fact, in {\sc Jade} every
node just stores a constant number of parameters, among them a fixed parameter
$\gamma$ that should be chosen so that the following main theorem holds:

\begin{theorem}\label{th:main1}
When running {\sc Jade} for $\Omega((T \log n)/\epsilon + (\log n)^4/(\gamma
\epsilon)^2)$ time steps, \textsc{Jade} has a constant competitive throughput
for any $(T,1-\epsilon)$-bounded adversary and any UDG w.h.p. as long as
$\gamma = O(1/(\log T+ \log \log n))$ and (a) the adversary is 1-uniform and
the UDG is connected, or (b) there are at least $2/\epsilon$ nodes within the
transmission range of every node.
\end{theorem}

Note that in practice, $\log T$ and $\log \log n$ are rather small so that our
condition on $\gamma$ is not too restrictive. Also, a conservative estimate on
$\log T$ and $\log \log n$ would leave room for a superpolynomial change in
$n$ and a polynomial change in $T$ over time.

On the other hand, we can also show the following result demonstrating that
Theorem~\ref{th:main1} essentially captures all the scenarios (within our
notation) under which {\sc Jade} can have a constant competitive throughput.

\begin{theorem} \label{th:main2}
If (a) the UDG is not connected, or (b) the adversary is allowed to
be 2-uniform and there are nodes with $o(1/\epsilon)$ nodes within
their transmission range, then there are cases in which {\sc Jade}
is not constant competitive for any constant $c$ independent of
$\epsilon$.
\end{theorem}

Certainly, no MAC protocol can guarantee a constant competitive throughput if
the UDG is not connected. However, it is still open whether there are simple
MAC protocols that are constant competitive under non-uniform jamming
strategies even if there are $o(1/\epsilon)$ nodes within the transmission
range of a node.

%

\section{Description of {\sc Jade}}\label{sec:algo}

This section first gives a short motivation for our algorithmic approach and
then presents the {\sc Jade} protocol in detail.

\subsection{Intuition}

The intuition behind our MAC protocol is simple: in each round,
each node $u$ tries to send a message with probability $p_u$ with
$p_u \le \hat{p}$ for some small constant $0<\hat{p}<1$. Consider
the unit disk $D(u)$ around node $u$ consisting of $u$'s
neighboring nodes as well as $u$.\footnote{In this paper, disks
(and later sectors) will refer both to 2-dimensional areas in the
plane as well as to the set of nodes in the respective areas. The
exact meaning will become clear in the specific context.}
Moreover, let $N(u)=D(u)\setminus \{u\}$ and $p=\sum_{v\in N(u)}
p_v$. Suppose that $u$ is sensing the channel. Let $q_0$ be the
probability that the channel is idle at $u$ and let $q_1$ be the
probability that exactly one node in $N(u)$ is sending a message.
%
%
It holds that $q_0 = \prod_{v\in N(u)} (1-p_v)$ and $q_1 =
\sum_{v\in N(u)} p_v \prod_{w \in N(u)\setminus\{v\}} (1-p_w)$.
Hence,
$$
q_1 \le \sum_{v\in N(u)} p_v \frac{1}{1-\hat{p}}
\prod_{w\in N(u)} (1-p_w)
  = \frac{q_0 \cdot p}{1-\hat{p}},~~
  q_1 \ge \sum_{v\in N(u)} p_v \prod_{w\in N(u)} (1-p_w) = q_0 \cdot p.
$$
Thus we have the following lemma, which has also been derived
in~\cite{singlehop08} for the single-hop case.

\begin{lemma}\label{cl_relation}
$q_0 \cdot p \le q_1 \le \frac{q_0}{1-\hat{p}} \cdot p$.
\end{lemma}

By Lemma \ref{cl_relation}, if a node $v$ observes that the number of rounds
in which the channel is idle is essentially equal  to the number of rounds in
which exactly one message is sent, then $p = \sum_{v\in N(v)}  p_v$ is likely
to be around 1 (if $\hat{p}$ is a sufficiently small constant), which would be
ideal.
%
%
%
Otherwise, the nodes know that they need to adapt their probabilities. Thus,
if we had sufficiently many cases in which an idle channel or exactly one
message transmission is observed (which is the case if the adversary does not
heavily jam the channel and $p$ is not too large), then one can adapt the
probabilities $p_v$ just based on these two events and ignore all cases in
which the wireless channel is blocked, either because the adversary is jamming
it or because at least two messages interfere with each other (see also
\cite{idlesense} for a similar conclusion). Unfortunately, $p$ can be very
high for some reason, which requires a more sophisticated strategy for
adjusting the access probabilities.

\subsection{Protocol Description}

In {\sc Jade}, each node $v$ maintains a probability value $p_v$, a threshold
$T_v$ and a counter $c_v$. The parameters $\hat{p}, \gamma > 0$ in the
protocol are fixed and the same for every node. $\hat{p}$ may be set to any
constant value so that $0<\hat{p} \le 1/24$, and $\gamma$ should be small
enough so that the condition in Theorem~\ref{th:main1} is met.


\begin{shaded}
Initially, every node $v$ sets $T_v:=1$, $c_v:=1$ and $p_v:=\hat{p}$.
Afterwards, the {\sc Jade} protocol works in synchronized
rounds.\cancel{\footnote{Recall that we assume synchronized rounds for the
analysis only, and a non-synchronized execution of the protocol would also
work as long as all nodes operate at roughly the same speed.}} In every round,
each node $v$ decides with probability $p_v$ to send a message. If it decides
not to send a message, it checks the following two conditions:
\begin{itemize}
\item If $v$ senses an idle channel, then
$p_v:= \min\{(1+\gamma)p_v, \hat{p}\}$.
\item If $v$ successfully receives a message, then
$p_v:=(1+\gamma)^{-1}p_v$ and $T_v:=\max\{T_v-1, 1\}$.
\end{itemize}
Afterwards, $v$ sets $c_v:=c_v+1$. If $c_v>T_v$ then it does the
following: $v$ sets $c_v:=1$, and if there was no round among the
past $T_v$ rounds in which $v$ sensed a successful message
transmission {\em or an idle channel}, then $p_v:= (1+\gamma)^{-1}
p_v$ and {\em $T_v:=\min\{T_v+1,2^{1/(4\gamma)}\}$ }.
\end{shaded}

As we will see in the upcoming section, the concept of using a
multiplicative-increase-multiplicative-decrease mechanism for $p_v$ and an
additive-increase-additive-decrease mechanism for $T_v$, as well as the slight
modifications of the protocol in \cite{singlehop08}, marked in italic above,
are crucial for {\sc Jade} to work.

\section{Analysis of {\sc Jade}}\label{sec:analysis}

In contrast the description of {\sc Jade}, its stochastic analysis is rather
involved as it requires to shed light onto the complex interplay of the nodes
all following their randomized protocol in a highly dependent manner. We first
prove Theorem~\ref{th:main1} (Sections~\ref{sec:upper} and
\ref{sec:contention}) and then prove Theorem~\ref{th:main2}
(Section~\ref{sec:lower}). In order to show the theorems, we will frequently
use the following variant of the Chernoff bounds \cite{singlehop08,SSS95}.

\begin{lemma}\label{lem_chernoff}
Consider any set of binary random variables $X_1,$ $\ldots,X_n$.
Suppose that there are values $p_1,\ldots,p_n \in [0,1]$ with $\E [
\prod_{i \in S} X_i] \le \prod_{i \in S} p_i$ for every set $S
\subseteq \{1,\ldots,n\}$. Then it holds for $X=\sum_{i=1}^n X_i$
and $\mu = \sum_{i=1}^n p_i$ and any $\delta>0$ that
\[
  \Pr[X \ge (1+\delta)\mu] \le \left( \frac{e^{\delta}}{(1+\delta)^{1+\delta}}
  \right)^{\mu} \le e^{-\frac{\delta^2 \mu}{2(1+\delta/3)}}.
\]
If, on the other hand, it holds that $\E [ \prod_{i \in S} X_i] \ge
\prod_{i \in S} p_i$ for every set $S \subseteq \{1,\ldots,n\}$,
then it holds for any $0 < \delta < 1$ that
\[
  \Pr[X \le (1-\delta)\mu] \le \left( \frac{e^{-\delta}}{(1-\delta)^{1-\delta}}
  \right)^{\mu} \le e^{-\delta^2 \mu / 2}.
\]
\end{lemma}

Throughout the section we assume that $\gamma=O(1/(\log T + \log \log n))$ is
sufficiently small.

\subsection{Proof of Theorem~\ref{th:main1}} \label{sec:upper}

We first look at a slightly weaker form of adversary. We say a round $t$ is
{\em open} for a node $v$ if $v$ and at least one other node in its
neighborhood are non-jammed (which implies that $v$'s neighborhood is
non-empty). An adversary is {\em weakly $(T,1-\epsilon)$-bounded} for some $T
\in \N$ and $0<\epsilon<1$ if the adversary is $(T,1-\epsilon)$-bounded and in
addition to this, at least a constant fraction of the non-jammed rounds at
each node are open in every time interval of size $w \ge T$.


\begin{theorem} \label{th:throughput}
When running \textsc{Jade} for $\Omega([T+(\log^3 n)/(\gamma^2\epsilon)] \cdot (\log n)/\epsilon)$
rounds it holds
w.h.p.~that {\sc Jade} is constant competitive for any weakly
$(T,1-\epsilon)$-bounded adversary.
\end{theorem}
\begin{proof}
First, we focus on a {\em time frame} $F$ consisting of $(\alpha \log
n)/\epsilon$ {\em subframes} of size $f = \alpha [T+(\log^3
n)/(\gamma^2\epsilon)]$
each, where $f$ is a multiple of $T$ and
$\alpha$ is a sufficiently large constant. The proof needs the following three
lemmas. The first one is identical to Claim 2.5 in \cite{singlehop08}. It is
true because only successful message transmissions reduce $T_u$.

\begin{lemma}\label{lem:claim2.5}
If in a time interval $I$ the number of rounds in which a node $u$
successfully receives a message is at most $r$, then $u$ increases $T_u$ in at
most $r + \sqrt{2|I|}$ rounds in $I$.
\end{lemma}

The second lemma holds for arbitrary (not just weakly)
$(T,1-\epsilon)$-bounded adversaries and will be shown in
Section~\ref{sec:contention}.

\begin{lemma}\label{lem:contention}
For every node $u$, $\sum_{v \in D(u)} p_v = O(1)$ for at least a $(1-\epsilon
\beta)$-fraction of the rounds in time frame $F$, w.h.p., where the constant
$\beta>0$ can be made arbitrarily small.
\end{lemma}

The third lemma just follows from some simple geometric argument.

\begin{lemma} \label{lem:region}
A disk of radius 2 can be cut into at most 20 regions so that the distance
between any two points in a region is at most 1.
\end{lemma}

Consider some fixed node $u$. Let $J \subseteq F$ be the set of all non-jammed
{\em open} rounds at $u$ in time frame $F$ (which are a constant fraction of
the non-jammed rounds at $u$ because we have a weakly $(T,
1-\epsilon)$-bounded adversary). Let $p$ be a constant satisfying
Lemma~\ref{lem:contention} (i.e., $\sum_{w \in D(v)} p_w \le p$). Define
$DD(u)$ to be the disk of radius 2 around $u$ (i.e., it has twice the radius
of $D(u)$). Cut $DD(u)$ into 20 regions $R_1,\ldots,R_{20}$ satisfying
Lemma~\ref{lem:region}, and let $v_i$ be any node in region $R_i$ (if such a
node exists), where $v_i=u$ if $u \in R_i$. According to
Lemma~\ref{lem:contention} it holds for each $i$ that at least a $(1-\epsilon
\beta'/20)$-fraction of the rounds in $F$ satisfy $\sum_{w \in D(v_i)} p_w \le
p$ for any constant $\beta'>0$, w.h.p. Thus, at least a $(1-\epsilon
\beta'')$-fraction of the rounds in $F$ satisfy $\sum_{w \in D(v_i)} p_w \le
p$ for {\em every} $i$ for any constant $\beta''>0$, w.h.p. As $D(v) \subseteq
DD(u)$ for all $v \in D(u)$ and $u$ has at least $\epsilon |F|$ non-jammed
rounds in $F$, we get the following lemma, which also holds for arbitrary
$(T,1-\epsilon)$-bounded adversaries

\begin{lemma}
At least a $(1-\beta)$-fraction of the rounds in $J$ satisfy $\sum_{v \in
D(u)} p_v \le p$ and $\sum_{w \in D(v)} p_w = O(p)$ for all nodes $v \in D(u)$
for any constant $\beta>0$, w.h.p.
\end{lemma}

Let us call these rounds {\em good}. Since the probability that $u$ senses the
channel is at least $1-\hat{p}$ and the probability that the channel at $u$ is
idle for $\sum_{w \in D(u)} p_w \le p$ is equal to $\prod_{v \in N(u)} (1-p_v)
\ge \prod_{v \in N(u)} e^{-2p_v} \ge e^{-2p}$, $u$ senses an idle channel for
at least
$
  (1-\hat{p})(1-\beta)|J| e^{-2p} \ge 2 \beta |J|
$
many rounds in $J$ on expectation if $\beta$ is sufficiently small. This also
holds w.h.p. when using the Chernoff bounds under the condition that at least
$(1-\beta)|J|$ rounds in $F$ are good (which also holds w.h.p.). Let $k$ be
the number of times $u$ receives a message in $F$. We distinguish between two
cases.

\medskip

\noindent {\em Case 1:}  $k \ge \beta |J|/6$. Then {\sc Jade} is constant
competitive for $u$ and we are done.

\medskip

\noindent {\em Case 2:} $k < \beta |J|/6$. Then we know from
Lemma~\ref{lem:claim2.5} that $p_u$ is decreased at most $\beta |J|/6 +
\sqrt{2|F|}$ times in $F$ due to $c_u > T_u$. In addition to this, $p_u$ is
decreased at most $\beta |J|/6$ times in $F$ due to a received message. On the
other hand, $p_u$ is increased at least $2\beta |J|$ times in $J$ (if
possible) due to an idle channel w.h.p. Also, we know from the {\sc Jade}
protocol that at the beginning of $F$, $p_u = \hat{p}$. Hence, there must be
at least
$
  \beta(2-1/6-1/6)|J| - \sqrt{2|F|} \ge (3/2)\beta |J|
$
rounds in $J$ w.h.p. at which $p_u = \hat{p}$. As there are at least
$(1-\beta)|J|$ good rounds in $J$ (w.h.p.), there are at least $\beta |J|/2$
good rounds in $J$ w.h.p. in which $p_u = \hat{p}$. For these good rounds, $u$
has a constant probability to transmit a message and every node $v \in D(u)$
has a constant probability of receiving it, so $u$ successfully transmits
$\Theta(|J|)$ messages to at least one of its non-jammed neighbors in $F$ (on
expectation and also w.h.p.).

\medskip

If we charge $1/2$ of each successfully transmitted message to the sender and
$1/2$ to the receiver, then a constant competitive throughput can be
identified for every node in both cases above, so {\sc Jade} is constant
competitive in $F$.

It remains to show that Theorem~\ref{th:throughput} also holds for larger time
intervals than $|F|$. First, note that all the proofs are valid as long as
$\gamma \le 1/[c(\log T + \log \log n)]$ for a constant $c \ge 2$, so we can
increase $T$ and thereby also $|F|$ as long as this inequality holds. So
w.l.o.g.~we may assume that $\gamma = 1/[2(\log T + \log \log n)]$. In this
case, $2^{1/(4\gamma)} \le \sqrt{|F|}$, so our rule of increasing $T_v$ in
{\sc Jade} implies that $T_v \le \sqrt{|F|}$ at any time. This allows us to
extend the competitive throughput result from a single to any sequence of
polynomial in $n$ many time frames $F$, which completes the proof of
Theorem~\ref{th:throughput}. \hfill \qed\end{proof}

Now, let us consider the two cases of Theorem~\ref{th:main1}. Recall that we
allow here any $(T,1-\epsilon)$-bounded adversary and not just a weakly
bounded.

\subsubsection*{Case 1: the adversary is 1-uniform and the UDG is connected.}

In this case, every node has a non-empty neighborhood and therefore {\em all}
non-jammed rounds of the nodes are open. Hence, the conditions on a weakly
$(T,1-\epsilon)$-bounded adversary are satisfied. So
Theorem~\ref{th:throughput} applies, which completes the proof of
Theorem~\ref{th:main1} a).

\subsubsection*{Case 2: {\boldmath $|D(v)| \ge 2/\epsilon$} for all {\boldmath $v \in V$}.}

Consider some fixed time interval $I$ with $|I|$ being a multiple of $T$. For
every node $v \in D(u)$ let $f_v$ be the number of non-jammed rounds at $v$ in
$I$ and $o_v$ be the number of open rounds at $v$ in $I$. Let $J$ be the set
of rounds in $I$ with at most one non-jammed node. Suppose that $|J| >
(1-\epsilon/2)|I|$. Then every node in $D(u)$ must have more than
$(\epsilon/2)|I|$ of its non-jammed rounds in $J$. As these non-jammed rounds
must be serialized in $J$ to satisfy our requirement on $J$, it holds that
$|J| > \sum_{v \in D(u)} (\epsilon/2)|I| \ge (2/\epsilon) \cdot
(\epsilon/2)|I| = |I|$. Since this is impossible, it must hold that $|J| \le
(1-\epsilon/2)|I|$.

Thus, $\sum_{v \in D(u)} o_v \ge (\sum_{v \in D(u)} f_v) - |J| \ge (1/2)
\sum_{v \in D(u)} f_v$ because $\sum_{v \in D(u)} f_v \ge (2/\epsilon) \cdot
\epsilon |I| =2|I|$. Let $D'(u)$ be the set of nodes $v \in D(u)$ with $o_v
\ge f_v/4$. That is, for each of these nodes, a constant fraction of the
non-jammed time steps is open. Then $\sum_{v \in D(u) \setminus D'(u)} o_v <
(1/4)\sum_{v \in D(u)} f_v$, so $\sum_{v \in D'(u)} o_v \ge (1/2) \sum_{v \in
D(u)} o_v\geq (1/4)\sum_{v \in D(u)} f_v$.

Consider now a set $U \subseteq V$ of nodes so that $\bigcup_{u
\in U} D(u) = V$ and for every $v \in V$ there are at most 6 nodes
$u \in U$ with $v \in D(u)$ ($U$ is easy to construct in a greedy
fashion for arbitrary UDGs and also known as a {\em dominating set
of constant density}). Let $V'=\bigcup_{u \in U} D'(u)$. Since
$\sum_{v \in D'(u)} o_v \ge (1/4) \sum_{v \in D(u)} f_v$ for every
node $u \in U$, it follows that $\sum_{v \in V'} o_v \ge (1/6)
\sum_{u\in U} \sum_{v \in D'(u)} o_v \ge (1/24) \sum_{u\in U}
\sum_{v \in D(u)} f_v \ge (1/24) \sum_{v \in V} f_v$. Using that
together with Theorem~\ref{th:throughput}, which implies that {\sc
Jade} is constant competitive w.r.t. the nodes in $V'$, completes
the proof of Theorem~\ref{th:main1} b).

\subsection{Proof of Lemma~\ref{lem:contention}} \label{sec:contention}

In order to finish the proof of Theorem~\ref{th:main1}, it remains to prove
Lemma~\ref{lem:contention}. Consider any fixed node $u$. We partition $u$'s
unit disk $D(u)$ into six \emph{sectors} of equal angles from $u$,
$S_1,...,S_6$. Note that all nodes within a sector $S_i$ have distances of at
most 1 from each other, so they can directly communicate with each other (in
$D(u)$, distances can be up to 2). We will first explore properties of an
arbitrary node in one sector, then consider the implications for a whole
sector, and finally bound the cumulative sending probability in the entire
unit disk.

Recall the definition of a time frame, a subframe and $f$ in the proof of
Theorem~\ref{th:throughput}. Fix a sector $S$ in $D(u)$ and consider some
fixed time frame $F$. Let us refer to the sum of the probabilities of the
neighboring nodes of a given node $v\in S$ by $\bar{p}_v := \sum_{w\in S
\setminus \{v\}} p_w $. The following lemma shows that $p_v$ will decrease
dramatically if $\bar{p}_v$ is high throughout a certain time interval.

\begin{lemma}\label{single_node}
Consider a node $v$ in a unit disk $D(u)$. If $\bar{p}_v>5-\hat{p}$ during
\emph{all} rounds of a subframe $I$ of $F$, then $p_v$ will be at most $1/n^2$
at the end of $I$, w.h.p.
\end{lemma}
\begin{proof}
 We say that a round is {\em useful} for node $v$ if
from $v$'s perspective there is an idle channel or a successful
transmission at that round (when ignoring the action of $v$);
otherwise the round is called {\em non-useful}. Note that in a
non-useful round, according to our protocol, $p_v$ will either
decrease (if the threshold $T_v$ is exceeded) or remain the same. On
the other hand, in a {\em useful} round, $p_v$ will increase (if $v$
senses an idle channel), decrease (if $v$ senses a successful
transmission) or remain the same (if $v$ sends a message). Hence,
$p_v$ can only increase during useful rounds of $I$. Let
$\mathcal{U}$ be the set of useful rounds in $I$ for our node $v$.
We distinguish between two cases, depending on the cardinality
$|\mathcal{U}|$. In the following, let $p_v(0)$ denote the
probability of $v$ at the beginning of $I$ (which is at most
$\hat{p}$). Suppose that $f \ge 2 [(3c \ln n)/\gamma]^2$ for a
sufficiently large constant $c$.
(This lower bound coincides with our definition of $f$ in the proof
of Theorem~\ref{th:throughput}.)

\emph{Case 1:} Suppose that $|\mathcal{U}| < (c \ln n)/\gamma$, that
is, many rounds are blocked and $p_v$ can increase only rarely. As
there are at least $(3c \ln n)/\gamma$ occasions in $I$ in which
$c_v>T_v$ and $|\mathcal{U}| < (c \ln n)/\gamma$, in at least $(2c
\ln n)/\gamma$ of these occasions $v$ only saw blocked channels for
$T_v$ consecutive rounds and therefore decides to increase $T_v$ and
decrease $p_v$. Hence, at the end of $I$,
\begin{eqnarray*}
  p_v \le (1+\gamma)^{|\mathcal{U}|-2c\ln n/\gamma} p_v(0)
   \le (1+\gamma)^{-c \ln n/\gamma} p_v(0)
   \le  e^{-c \ln n} = 1/n^c .
\end{eqnarray*}

\emph{Case 2:} Next, suppose that $|\mathcal{U}| \ge (c \ln
n)/\gamma$. We will show that many of these useful rounds will be
successful such that $p_v$ decreases. Since $p_v \le \hat{p} \le
1/24$ throughout $I$, it follows from the Chernoff bounds that
w.h.p.~$v$ will sense the channel for at least a fraction of $2/3$
of the useful rounds w.h.p. Let this set of useful rounds be called
$\mathcal{U}'$. Consider any round $t \in \mathcal{U}'$. Let $q_0$
be the probability that there is an idle channel at round $t$ and
$q_1$ be the probability that there is a successful transmission at
$t$. It holds that $q_0+q_1 = 1$. From Lemma~\ref{cl_relation} we
also know that $q_1 \ge q_0 \cdot \bar{p}_v$. Since $\bar{p}_v >
5-\hat{p}$ for all rounds in $I$, it follows that $q_1 \ge 4/5$ for
every round in $\mathcal{U}'$. Thus, it follows from the Chernoff
bounds that for at least $2/3$ of the rounds in $\mathcal{U}'$, $v$
will sense a successful transmission w.h.p. Hence, at the end of $I$
it holds w.h.p.~that
\begin{eqnarray*}
 p_v \le (1+\gamma)^{-(1/3) \cdot |\mathcal{U}'|} p_v(0)
     \le (1+\gamma)^{-(1/3) \cdot (2c/3) \ln n / \gamma} p_v(0)
      \le  e^{-(2c/9) \ln n} = 1/n^{2c/9} .
\end{eqnarray*}
Combining the two cases with $c \ge 9$ results in the lemma. \hfill
\qed
\end{proof}

Given this property of the individual probabilities, we can derive a
bound for the cumulative probability of an entire sector $S$. In
order to compute $p_S=\sum_{v\in S}p_v$, we introduce three
thresholds, a low one, $\rho_{green}=5$, one in the middle,
$\rho_{yellow}=5e$, and a high one, $\rho_{red}=5e^2$. The following
three lemmas provide some important insights about these
probabilities.
\begin{lemma}\label{fast:recovery}
For any subframe $I$ in $F$ and any initial value of $p_S$ in $I$ there is at
least one round in $I$ with $p_S \le \rho_{green}$ w.h.p.
\end{lemma}
\begin{proof}
We prove the lemma by contradiction. Suppose that throughout the
entire interval $I$, $p_S > \rho_{green}$. Then it holds for every
node $v \in S$ that $\bar{p}_v > \rho_{green} - \hat{p}$ throughout
$I$. In this case, however, we know from Lemma~\ref{single_node},
that $p_v$ will decrease to at most $1/n^2$ at the end of $I$ w.h.p.
Hence, all nodes $v \in S$ would decrease $p_v$ to at most $1/n^2$
at the end of $I$ w.h.p., which results in $p_S \le 1/n$. This
contradicts our assumption, so w.h.p.~there must be a round $t$ in
$I$ at which $p_S\le \rho_{green}$. \hfill \qed\end{proof}

\begin{lemma}\label{lemma:helper}
For any time interval $I$ in $F$ of size $f$ and any sector $S$ it holds that
if $p_S \le \rho_{green}$ at the beginning of $I$, then $p_S \le
\rho_{yellow}$ throughout $I$, w.m.p. Similarly, if $p_S \le \rho_{yellow}$ at
the beginning of $I$, then $p_S \le \rho_{red}$ throughout $I$, w.m.p.
\end{lemma}
\begin{proof}
It suffices to prove the lemma for the case that initially $p_S \le
\rho_{green}$ as the other case is analogous. Consider some fixed
round $t$ in $I$. Let $p_S$ be the cumulative probability at the
beginning of $t$ and $p'_S$ be the cumulative probability at the end
of $t$. Moreover, let $p_{S}^{(0)}$ denote the cumulative
probability of the nodes $w \in S$ with no transmitting node in
$D(w) \setminus S$ in round $t$. Similarly, let $p_{S}^{(1)}$ denote
the cumulative probability of the nodes $w \in S$ with a single
transmitting node in $D(w) \setminus S$, and let $p_{S}^{(2)}$ be
the cumulative probability of the nodes $w \in S$ that experience a
blocked round either because they are jammed or at least two nodes
in $D(w) \setminus S$ are transmitting at $t$. Certainly, $p_S =
p_S^{(0)}+p_S^{(1)}+p_S^{(2)}$. Our goal is to determine $p'_S$ in
this case. Let $q_0(S)$ be the probability that all nodes in $S$
stay silent, $q_1(S)$ be the probability that exactly one node in
$S$ is transmitting, and $q_2(S) = 1-q_0(S)-q_1(S)$ be the
probability that at least two nodes in $S$ are transmitting.

%
When ignoring the case that $c_v > T_v$ for a node $v \in S$ at
round $t$, it holds:
\begin{eqnarray*}
\E[p'_S] &=& q_0(S) \cdot [(1+\gamma) p_S^{(0)} + (1+\gamma)^{-1}
p_S^{(1)} + p_S^{(2)}] \\
        && +
             q_1(S) \cdot [(1+\gamma)^{-1} p_S^{(0)} + p_S^{(1)}+p_S^{(2)}] \\
        && + q_2(S) \cdot [p_S^{(0)} + p_S^{(1)} + p_S^{(2)}] \\
\end{eqnarray*}
This is certainly also an upper bound for $\E[p'_S]$ if $c_v > T_v$
for a node $v \in S$ because $p_v$ will never be increased (but
possibly decreased) in this case. Now, consider the event $E_2$ that
at least two nodes in $S$ transmit a message. If $E_2$ holds, then
$\E[p'_S] = p'_S = p_S$, so there is no change in the system. On the
other hand, assume that $E_2$ does not hold. Let
$q'_0(S)=q_0(S)/(1-q_2(S))$ and $q'_1(S)=q_1(S)/(1-q_2(S))$ be the
probabilities $q_0(S)$ and $q_1(S)$ under the condition of $\neg
E_2$. Then we distinguish between three cases.

\noindent {\it Case 1:} $p_S^{(0)}=p_S$. Then
\begin{eqnarray*}
  \E[p'_S] & \le & q'_0(S) \cdot (1+\gamma) p_S + q'_1(S) \cdot (1+\gamma)^{-1} p_S \\
  & = & ((1+\gamma)q'_0(S) + (1+\gamma)^{-1}q'_1(S)) p_S .
\end{eqnarray*}
From Lemma \ref{cl_relation} we know that $q_0(S) \le q_1(S) / p_S$,
so $q'_0(S) \le q'_1(S) / p_S$. If $p_S \ge \rho_{green}$, then
$q'_0(S) \le q'_1(S)/5$. Hence,
\begin{eqnarray*}
  \E[p'_S] & \le & ((1+\gamma)/6 + (1+\gamma)^{-1}5/6) p_S \le (1+\gamma)^{-1/2} p_S
\end{eqnarray*}
since $\gamma=o(1)$. On the other hand, $p'_S \le (1+\gamma) p_S$ in
any case.

\medskip

\noindent {\it Case 2}: $p_S^{(1)}=p_S$. Then
\begin{eqnarray*}
  \E[p'_S] & \le & q'_0(S) \cdot (1+\gamma)^{-1} p_S + q'_1(S) p_S \\
  & = & (q'_0(S)/(1+\gamma) + (1-q'_0(S))) p_S = (1 - q'_0(S)
  \gamma/(1+\gamma))p_S .
\end{eqnarray*}
Now, it holds that $1-x \gamma/(1+\gamma) \le (1+\gamma)^{-x/2}$ for
all $x \in [0,1]$ because from the Taylor series of $e^x$ and
$\ln(1+x)$ it follows that
\[
  (1+\gamma)^{-x/2} \ge 1- (x \ln (1+\gamma))/2 \ge 1-(x (1-\gamma/2)\gamma)/2
\]
and
\[
  1 - x \gamma/(1+\gamma) \le 1-(x (1-\gamma/2)\gamma)/2
\]
for all $x, \gamma \in [0,1]$ as is easy to check. Therefore, when
defining $\varphi = q'_0(S)$, we get $\E[p'_S] \le
(1+\gamma)^{-\varphi/2} p_S$. On the other hand, $p'_S \le p_S \le
(1+\gamma)^{\varphi} p_S$.

\medskip

\noindent {\it Case 3}: $p_S^{(2)}=p_S$. Then for $\varphi=0$,
$\E[p'_S] \le p_S = (1+\gamma)^{-\varphi/2} p_S$ and $p'_S \le p_S =
(1+\gamma)^{\varphi} p_S$.

\medskip

Combining the three cases and taking into account that
$p_S^{(0)}+p_S^{(1)}+p_S^{(2)}=p_S$, we obtain the following result.

\begin{lemma}
There is a $\phi \in [0,1]$ (depending on $p_S^{(0)}$, $p_S^{(1)}$
and $p_S^{(2)}$) so that
\begin{eqnarray}
  \E[p'_S] \le (1+\gamma)^{-\phi} p_S \quad \mbox{and} \quad p'_S \le (1+\gamma)^{2\phi} p_S .
  \label{eq:pS}
\end{eqnarray}
\end{lemma}
\begin{proof}
Let $a=(1+\gamma)^{1/2}$, $b=(1+\gamma)^{\varphi/2}$ for the
$\varphi$ defined in Case 2, and $c=1$. Furthermore, let $x_0 =
p_S^{(0)}/p_S$, $x_1 = p_S^{(1)}/p_S$ and $x_2 =p_S^{(2)}/p_S$.
Define $\phi = - \log_{1+\gamma} ((1/a)x_0 + (1/b)x_1 + (1/c)x_2)$.
Then we have
\[
  \E[p'_S] \le (1+\gamma)^{-1/2} p_S^{(0)} + (1+\gamma)^{-\varphi/2} p_S^{(1)} + p_S^{(2)}
  = (1+\gamma)^{-\phi} p_S .
\]
We need to show that for this $\phi$, also $p'_S \le
(1+\gamma)^{2\phi} p_S$. As $p'_S \le (1+\gamma) p_S^{(0)} +
(1+\gamma)^{\varphi} p_S^{(1)} + p_S^{(2)}$, this is true if
\[
  a^2 x_0 + b^2 x_1 + c^2 x_2 \le \frac{1}{((1/a)x_0 + (1/b)x_1 + (1/c)x_2)^2}
\]
or
\begin{eqnarray}
  ((1/a)x_0 + (1/b)x_1 + (1/c)x_2)^2 (a^2 x_0 + b^2 x_1 + c^2 x_2) \le 1
  \label{eq:helper}
\end{eqnarray}
To prove this, we need two claims whose proofs are tedious but
follow from standard math.

\begin{claim}
For any $a,b,c>0$ and any $x_0,x_1,x_2>0$ with $x_0+x_1+x_2=1$,
\[
  (a x_0 + b x_1 + c x_2)^2 \le (a^2 x_0 + b^2 x_1 + c^2 x_2)
\]
\end{claim}

\begin{claim}
For any $a,b,c>0$ and any $x_0,x_1,x_2>0$ with $x_0+x_1+x_2=1$,
\[
  ((1/a)x_0 + (1/b)x_1 + (1/c)x_2)(a x_0 + b x_1 + c x_2) \le 1
\]
\end{claim}

Combining the claims, Equation~\ref{eq:helper} follows, which
completes the proof. \hfill \qed\end{proof}

Hence, for any outcome of $E_2$, $\E[p'_S] \le (1+\gamma)^{-\varphi}
p_S$ and $p'_S \le (1+\gamma)^{2\varphi} p_S$ for some $\varphi \in
[0,1]$. If we define $q_S = \log_{1+\gamma} p_S$, then it holds that
$\E[q'_S] \le q_S - \varphi$. For any time $t$ in $I$, let $q_t$ be
equal to $q_S$ at time $t$ and $\varphi_t$ be defined as $\varphi$
at time $t$. Our calculations above imply that as long as $p_S \in
[\rho_{green}, \rho_{yellow}]$, $\E[q_{t+1}] \le q_t - \varphi_t$
and $q_{t+1} \le q_t + 2\varphi_t$.

Now, suppose that within subframe $I$ we reach a point $t$ when $p_S
> \rho_{yellow}$. Since we start with $p_S \le \rho_{green}$, there
must be a time interval $I' \subseteq I$ so that right before $I'$,
$p_S \le \rho_{green}$, during $I'$ we always have $\rho_{green} <
p_S \le \rho_{yellow}$, and at the end of $I'$, $p_S >
\rho_{yellow}$. We want to bound the probability for this to happen.

Consider some fixed interval $I'$ with the properties above, i.e.,
with $p_S \le \rho_{green}$ right before $I'$ and $p_S \ge
\rho_{green}$ at the first round of $I'$, so initially, $p_S \in
[\rho_{green}, (1+\gamma)\rho_{green}]$. We use martingale theory to
bound the probability that in this case, the properties defined
above for $I'$ hold. Consider the rounds in $I'$ to be numbered from
1 to $|I'|$, let $q_t$ and $\varphi_t$ be defined as above, and let
$q'_t = q_t + \sum_{i=1}^{t-1} \varphi_i$. It holds that
\begin{eqnarray*}
  \E[q'_{t+1}] = \E[q_{t+1} + \sum_{i=1}^{t} \varphi_i] = \E[q_{t+1}] + \sum_{i=1}^{t}\varphi_i
  \le q_t - \varphi_t + \sum_{i=1}^{t}\varphi_i
  = q_t + \sum_{i=1}^{t-1} \varphi_i
  = q'_t.
\end{eqnarray*}
Moreover, it follows from Inequality (\ref{eq:pS}) that for any
round $t$, $p'_S \le (1+\gamma)^{2\varphi_t} p_S$. Therefore,
$q_{t+1} \le q_t + 2\varphi_t$, which implies that $q'_{t+1} \le
q'_t + \varphi_t$. Hence, we can define a martingale $(X_t)_{t \in
I'}$ with $\E[X_{t+1}] = X_t$ and $X_{t+1} \le X_t + \varphi_t$ that
stochastically dominates $q'_t$. Recall that a random variable $Y_t$
{\em stochastically dominates} a random variable $Z_t$ if for any
$z$, $\Pr[Y_t \ge z] \ge \Pr[Z_t \ge z]$. In that case, it is also
straightforward to show that $\sum_i Y_i$ stochastically dominates
$\sum_i Z_i$, which we will need in the following. Let $T = |I'|$.
We will make use of Azuma's inequality to bound $X_T$.

\begin{fact}[Azuma Inequality]
Let $X_0, X_1, \ldots$ be a martingale satisfying the property that
$X_i \le X_{i-1} + c_i$ for all $i\ge 1$. Then for any $\delta \ge
0$,
\[
  \Pr[X_T > X_0 + \delta] \le e^{-\delta^2/(2 \sum_{i=1}^T c_i^2)}.
\]
\end{fact}

Thus, for $\delta = 1/\gamma + \sum_{i=1}^T \varphi_i$ it holds in
our case that
\[
  \Pr[X_T > X_0 + \delta] \le e^{-\delta^2/(2 \sum_{i=1}^T
  \varphi_i^2)}.
\]
This implies that
\[
  \Pr[q'_T > q'_0 + \delta] \le e^{-\delta^2/(2 \sum_{i=1}^T
  \varphi_i^2)},
\]
for several reasons. First of all, stochastic dominance holds as
long as $p_S \in [\rho_{green},\rho_{yellow}]$, and whenever this is
violated, we can stop the process as the requirements on $I'$ would
be violated, so we would not have to count that probability towards
$I'$. Therefore,
\[
  \Pr[q_T > q_0 + 1/\gamma] \le e^{-\delta^2/(2 \sum_{i=1}^T
  \varphi_i^2)}.
\]
Notice that $q_T > q_0 + 1/\gamma$ is required so that $p_S >
\rho_{yellow}$ at the end of $I'$, so the probability bound above is
exactly what we need. Let $\varphi = \sum_{i=1}^T \varphi_i$. Since
$\varphi_i \le 1$ for all $i$, $\varphi \ge \sum_{i=1}^T
\varphi_i^2$. Hence,
\begin{eqnarray*}
  \frac{\delta^2}{2 \sum_{i=1}^T \varphi_i^2} & \ge & \frac{(1/\gamma +
    \varphi)^2}{2 \varphi}
   \ge  \left( \frac{1}{2\varphi \gamma^2} + \frac{\varphi}{2}
   \right).
\end{eqnarray*}
This is minimized for $1/(2\varphi \gamma^2) = \varphi/2$ or
equivalently, $\varphi=1/\gamma$. Thus,
\[
  \Pr[q_T > q_0 + 1/\gamma] \le e^{-1/\gamma}
\]
Since there are at most ${f \choose 2}$ ways of selecting $I'
\subseteq I$, the probability that there exists an interval $I'$
with the properties above is at most
\[
  {f \choose 2} e^{-1/\gamma} \le f^2 e^{-1/\gamma} \le \frac{1}{\log^c n}
\]
for any constant $c$ if $\gamma = O(1/(\log T + \log \log n))$ is
small enough. \hfill \qed \end{proof}


\begin{lemma}\label{sector}
For any subframe $I$ in $F$ it holds that if there has been at least one round
during the past subframe where $p_S \le \rho_{green}$, then throughout $I$,
$p_S \le \rho_{red}$ w.m.p.
\end{lemma}
\begin{proof}
Suppose that there has been at least one round during the past
subframe where $p_S \le \rho_{green}$. Then we know from
Lemma~\ref{lemma:helper} that w.m.p. $p_S \le \rho_{yellow}$ at the
beginning of $I$. But if $p_S \le \rho_{yellow}$ at the beginning of
$I$, we also know from Lemma~\ref{lemma:helper} that w.m.p. $p_S \le
\rho_{red}$ throughout $I$, which proves the lemma. \hfill
\qed\end{proof}

Now, define a subframe $I$ to be {\em good} if $p_S \le \rho_{red}$ throughout
$I$, and otherwise $I$ is called {\em bad}. With the help of
Lemma~\ref{fast:recovery} and Lemma~\ref{sector} we can prove the following
lemma.

\begin{lemma}\label{lem:upper}
For any sector $S$, at most $\epsilon \beta/6$ of the subframes $I$ in $F$ are
bad w.h.p., where the constant $\beta>0$ can be made arbitrarily small
depending on the constant $\alpha$ in $f$.
\end{lemma}
\begin{proof}
From Lemma~\ref{fast:recovery} it follows that for every subframe
$I$ in $F$ there is a time point $t \in I$ at which $p_S \le
\rho_{green}$ w.h.p. Consider now some fixed subframe $I$ in $F$
that is not the first one and suppose that the previous subframe in
$F$ had at least one round with $p_S \le \rho_{green}$. Then it
follows from Lemma~\ref{sector} that for all rounds in $I$, $p_S \le
\rho_{red}$ w.m.p. (where the probability only depends on $I$ and
its preceding subframe), i.e., $I$ is good. Hence, it follows from
the Chernoff bounds that at most $\epsilon \beta/7$ of the
odd-numbered as well as the even-numbered subframes after the first
subframe in $F$ are bad w.h.p.~(if the constant $\alpha$ is
sufficiently large). This implies that overall at most $\epsilon
\beta/6$ of the subframes in $F$ are bad w.h.p. \hfill \qed
\end{proof}

From Lemma~\ref{lem:upper} it follows that apart from an $\epsilon
\beta$-fraction of the subframes, all subframes $I$ in $F$ satisfy $\sum_{v
\in D(u)} p_v \in O(1)$ throughout $I$, which completes the proof of
Lemma~\ref{lem:contention}.

\cancel{
\subsection{Self-Stabilization}

Note that if the protocol is well-initialized, i.e., initially $T_v=1$ for all
$v$, then
the minimum time that is sufficient for
Theorem~\ref{th:main1} is $\Theta((T+(\log^3 n)/(\epsilon \gamma^2)) \cdot
(\log n)/\epsilon)$. Notice that $\hat{T} := \max_v T_v$ is bounded over time
as shown in the next lemma, so {\sc Jade} does not run into significant aging
problems.

\begin{lemma}
In each of the first polynomial in $n$ and $T$ many rounds, $\hat{T}
\in T+ n^{O(\gamma)}$ w.h.p.
\end{lemma}
\begin{proof}
From the proof of Lemma~\ref{lemma:helper} it follows that for every
node $u$, $\sum_{v \in D(u)} p_v = O(\gamma \ln n)$ for every round
(out of polynomial in $n$ and $T$ many rounds) w.h.p. This implies
that the probability that $u$ sees an idle channel is $e^{-O(\gamma
\ln n)} = 1/n^{O(\gamma)}$. Thus, it follows from the Chernoff
bounds and {\sc Jade} that within $T+n^{O(\gamma)}$ steps $u$ will
sense an idle channel at least once w.h.p. Since $T_u$ is only
increased if all prior $T_u$ time steps were blocked from the
perspective of $u$, $T_u$ is limited to $T+n^{O(\gamma)}$. \hfill
\qed\end{proof} }

\subsection{Limitations of the {\sc Jade} Protocol} \label{sec:lower}

One may ask whether a stronger throughput result than Theorem~\ref{th:main1}
can be shown. Ideally, we would like to use the following model. A MAC
protocol is called {\em strongly $c$-competitive} against some
$(T,1-\epsilon)$-bounded adversary if, for any sufficiently large time
interval and any node $v$, the number of rounds in which $v$ successfully
receives a message is at least a $c$-fraction of the total number of
non-jammed rounds at $v$. In other words, a strongly $c$-competitive MAC
protocol can achieve at least a $c$-fraction of the best possible throughput
for every individual node. Unfortunately, such a protocol seems to be
difficult to design. In fact, {\sc Jade} is not strongly $c$-competitive for
any constant $c>0$, even if the node density is sufficiently high.

\begin{theorem}
In general, {\sc Jade} is not strongly $c$-competitive for a constant $c>0$ if
the adversary is allowed to be 2-uniform and $\epsilon \le 1/3$.
\end{theorem}
\begin{proof}
Suppose that (at some corner of the UDG) we have a set $U$ of at least
$1/\hat{p}$ nodes located closely to each other that are all within the
transmission range of a node $v$. Initially, we assume that $\sum_{u \in U}
p_u \ge 1$, $p_v = \hat{p}$ and $T_x = 1$ for all nodes $x \in U \cup \{v\}$.
The time is partitioned into time intervals of size $T$. In each such time
interval, called {\em $T$-interval}, the $(T,1-\epsilon)$-bounded adversary
jams all but the first $\epsilon T$ rounds at $U$ and all but the last
$\epsilon T$ rounds at $v$. It follows directly from Section 2.3 of
\cite{singlehop08} that if $T=\Omega((\log^3 n)/(\gamma^2 \epsilon))$, then
for every node $u \in U$, $T_u \le \alpha \sqrt{T \log n / \epsilon}$ w.h.p.
for some sufficiently large constant $\alpha$. Thus, $T_u \le \gamma T/(\beta
\log n)$ w.h.p. for any constant $\beta>0$ if $T$ is sufficiently large.
Hence, between the last non-jammed round at $U$ and the first non-jammed round
at $v$ in a $T$-interval, the values $T_u$ are increased (and the values $p_u$
are decreased) at least $\beta(\log n)/(6\gamma)$ times. Thus, at the first
non-jammed round at $v$, it holds for every $u \in U$ that
\[
  p_u \le \hat{p} \cdot (1+\gamma)^{-\beta(\log n)/(6\gamma)}
  \le \hat{p} \cdot e^{-(\beta/6) \log n} \le 1/n^{\beta/6}
\]
and, therefore, $\sum_{u \in U} p_u = O(1/n^2)$ if $\beta \ge 18$. This
cumulative probability will stay that low during all of $v$'s non-jammed
rounds as during these rounds the nodes in $U$ are jammed. Hence, the
probability that $v$ receives any message during its non-jammed rounds of a
$T$-interval is $O(1/n^2)$, so {\sc Jade} is not $c$-competitive for $v$ for
any constant $c>0$. \hfill \qed\end{proof}

\begin{figure} [ht]
\begin{center}
\includegraphics[width=0.495\columnwidth]{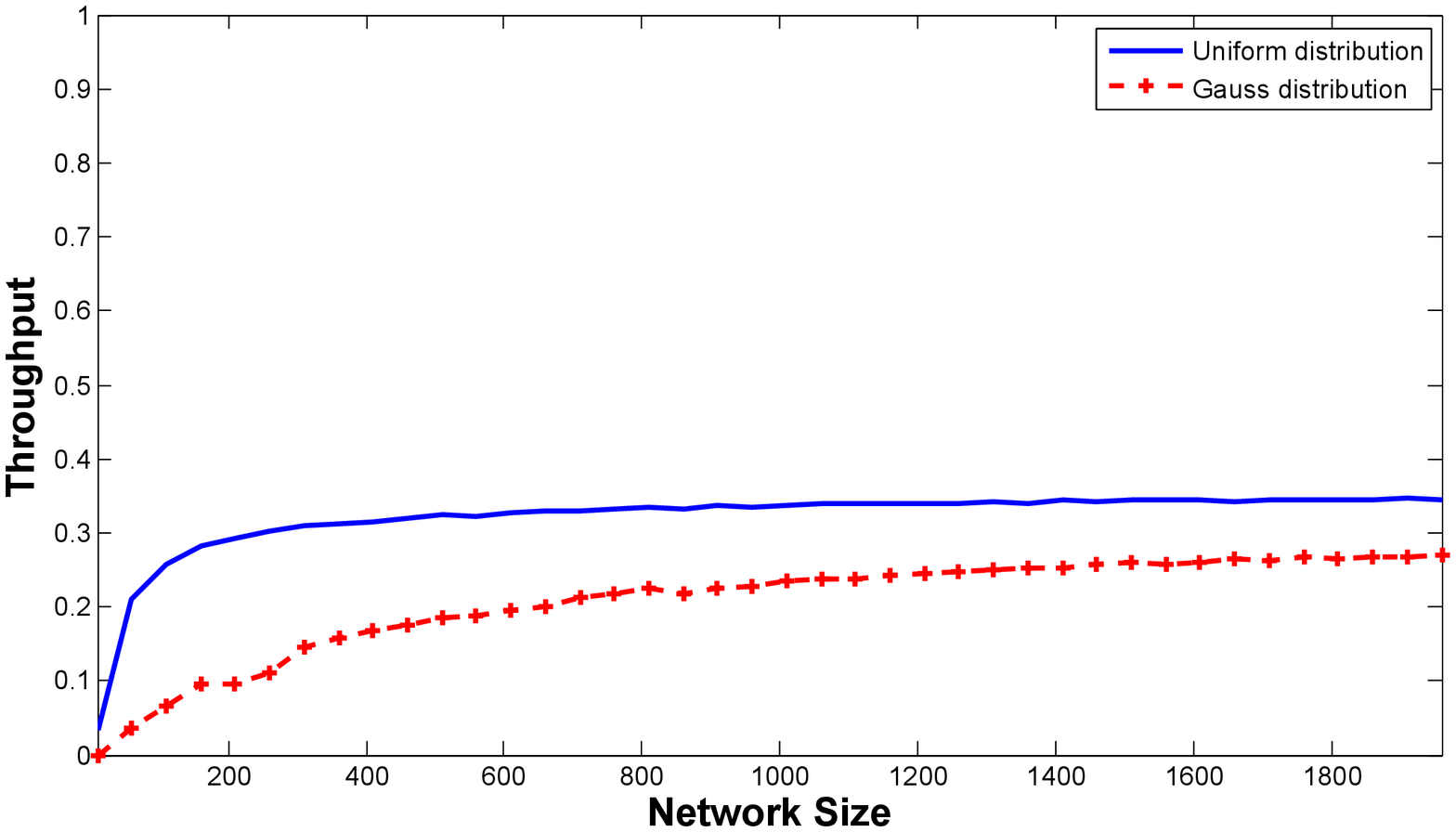}\\
\includegraphics[width=0.495\columnwidth]{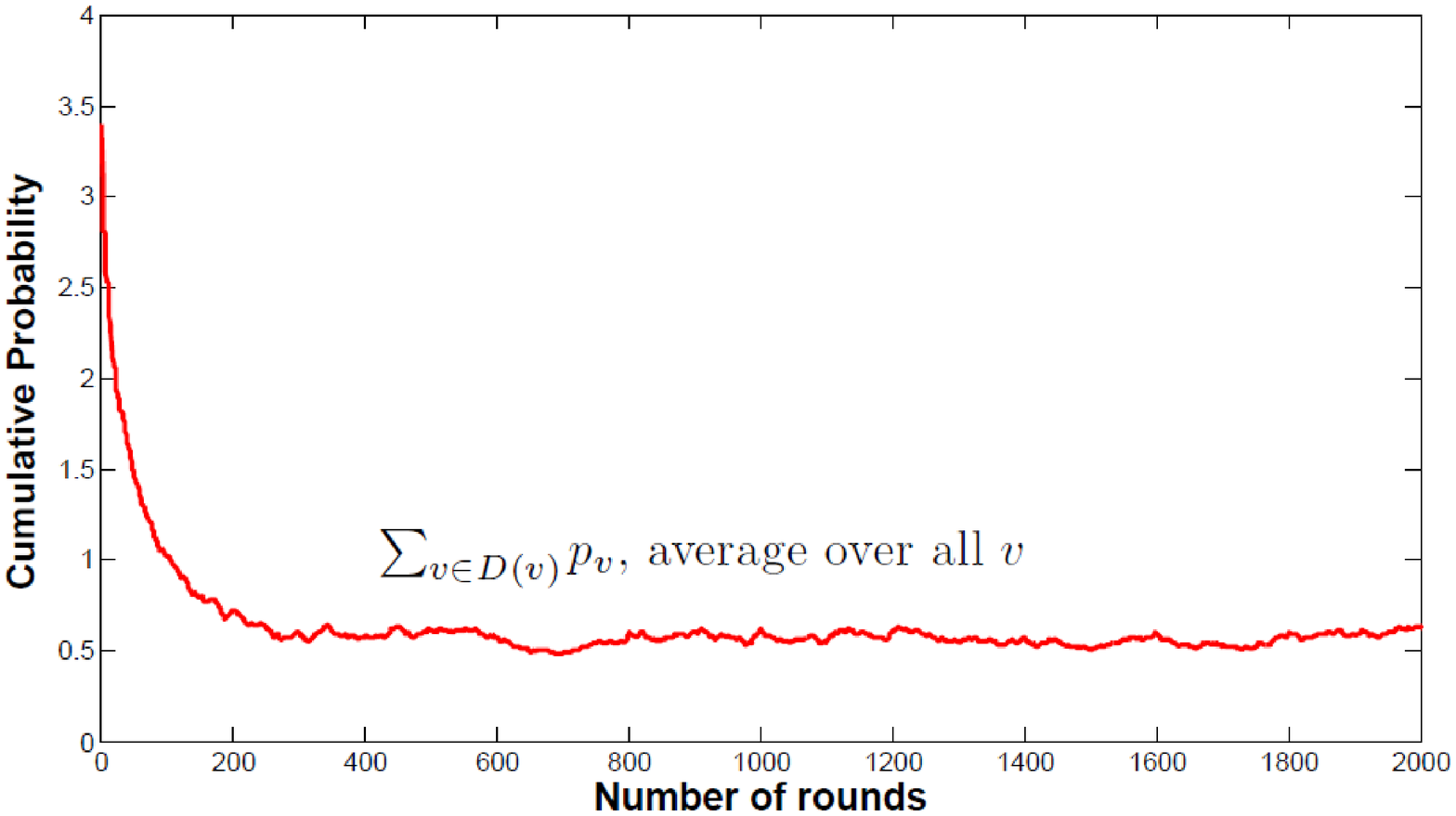}~\includegraphics[width=0.495\columnwidth]{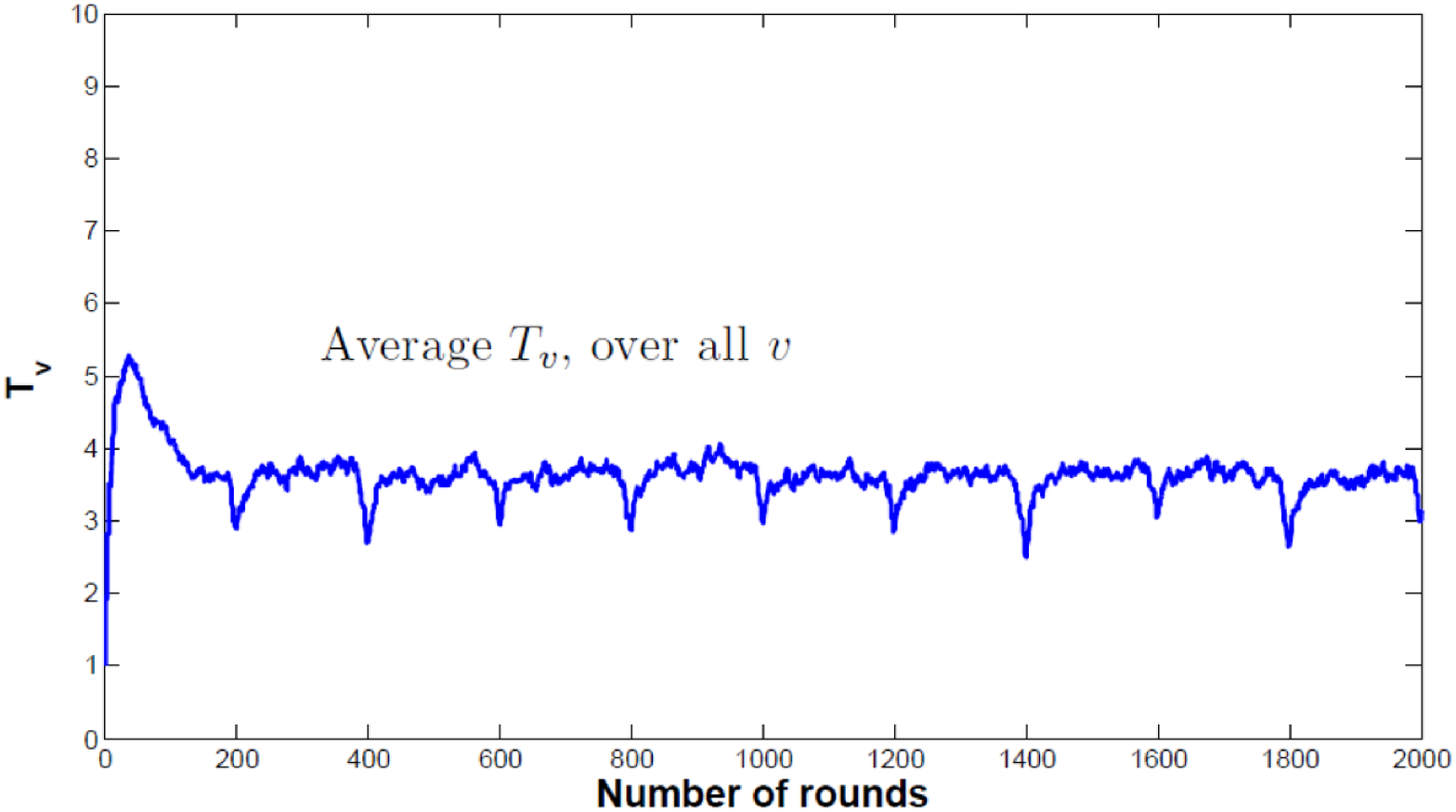}\\
\caption{\emph{Top:} Throughput as a function of network size.
\emph{Bottom left:} Convergence behavior for multi-hop networks
(uniform distribution). For the plot, we used $n=500$. Note that
the start-up phase where the sending probabilities are high is
short (no more than 50 rounds). \emph{Bottom right:} Convergence
of $T_v$ for multi-hop networks (uniform distribution). For the
plot, we used $n=500$.}\label{fig:succ}
\end{center}
\end{figure}

Also, in our original model, {\sc Jade} is not constant
competitive if the node density is too low.

\begin{theorem}
In general, {\sc Jade} is not $c$-competitive for a constant $c$ independent
of $\epsilon$ if there are nodes $u$ with $|D(u)|=o(1/\epsilon)$ and the
adversary is allowed to be 2-uniform.
\end{theorem}
\begin{proof}
Suppose that we have a set $U$ of $k=o(1/\epsilon)$ nodes located closely to
each other that are all within the transmission range of a node $v$. Let
$T=\Omega((\log^3 n)/(\gamma^2 \epsilon))$. In each $T$-interval, the
adversary never jams $v$ but jams all but the first $\epsilon T$ rounds at
$U$. Then Section 2.3 of \cite{singlehop08} implies that for every node $u \in
U$, $T_u \le \gamma T/(\beta \log n)$ w.h.p. for any constant $\beta>0$ if $T$
is sufficiently large. The nodes in $U$ continuously increase their
$T_u$-values and thereby reduce their $p_u$ values during their jammed time
steps. Hence, the nodes in $U \cup \{v\}$ will receive at most $\epsilon T
\cdot |U| + (\epsilon T + O(T/\log n)) = \epsilon T \cdot o(1/\epsilon) +
(\epsilon + o(1))T = (\epsilon +o(1))T$ messages in each $T$-interval on
expectation whereas the sum of non-jammed rounds over all nodes is more than
$T$. \hfill \qed\end{proof}

This implies Theorem~\ref{th:main2}. Hence, Theorem~\ref{th:main1} is
essentially the best one can show for {\sc Jade} (within our notation).

\subsection{Simulations} \label{sec:simulation}

In order to complement our theoretical insights, we conducted some
experiments. First, we present our throughput results for a
sufficiently large time interval, and then we discuss the
convergence behavior. For our simulations, as in our formal
analysis, we assume that initially all nodes $v\in V$ have a high
sending probability of $p_v = \hat{p}=1/24$. The nodes are
distributed at random over a square plane of $4\times 4$ units, and
are connected in a unit disk graph manner (multi-hop). We simulate
the jamming activity in the following way: for each round, a node is
jammed independently with probability $(1-\epsilon)$.
We run the simulation for a sufficiently large number of time steps
indicated by the Theorem~\ref{th:main1}, i.e., for $([T+(\log^3
n)/(\gamma^2 \epsilon)] \cdot (\log n)/\epsilon$
rounds, where $\epsilon=0.3$, $T=200$, and $\gamma =0.1$.

Figure~\ref{fig:succ}~(\emph{top}) shows the throughput
competitiveness of {\sc Jade} for a scenario where different numbers
of nodes are distributed uniformly at random over the plane and a
scenario where the nodes are distributed according to a
normal/Gaussian distribution $\mathcal{N}(0,1)$. In both cases, the
throughput is larger when the density is higher. This corresponds to
our formal insight that a constant competitive throughput is
possible only if the node density exceeds a certain threshold. For
example, this holds in case there are $60$ nodes in the $4\times 4$
plane (density of $3.75$), as there are at least $3.75 \pi \approx
12 > 2/\epsilon \approx 7$ uniformly distributed nodes in one unit
disk. As can be seen in the figure, when the number of nodes is
larger than $60$, the throughput falls in a range between $20\%$ and
$35\%$.


Convergence time is the second most important evaluation criterion.
We found that already after a short time, a constant throughput is
achieved; in particular, the total sending probability per unit disk
approaches a constant value quickly. This is due to the nodes'
ability to adapt their sending probabilities fast, see
Figure~\ref{fig:succ}~(\emph{bottom left}). The figure also
illustrates the high correlation between success ratio and
aggregated sending probability.


Finally, we have also studied the average of the $T_v$ values over
time. While initially, due to the high sending probabilities, the
$T_v$ intervals are large (up to around $5$ if $n=500$), they
decline quickly, similarly to our observations made in the previous
plots. The average of $T_v$ values stabilize in an interval $[2,4]$,
as shown in Figure~\ref{fig:succ}~(\emph{bottom right}).

\section{Conclusion}\label{sec:conclusion}

This paper has presented the first jamming-resistant MAC protocol
with provably good performance in multi-hop networks exposed to an
adaptive adversary. While we have focused on unit disk graphs, we
believe that our stochastic analysis is also useful for more
realistic wireless network models. Moreover, although our analysis
is involved, our protocol is rather simple.



{\footnotesize
  \bibliographystyle{abbrv} \bibliography{jammers}
}

\end{document}